\documentclass[letterpaper, 10 pt, conference]{ieeeconf} 
\IEEEoverridecommandlockouts

\usepackage{enumitem}
\setlist{topsep=0pt, leftmargin=*}
\usepackage{algorithm}
\usepackage{algorithmicx}

\usepackage[noend]{algpseudocode}
\usepackage{amsmath, amsfonts, bm, amssymb, tabularx}
\usepackage{latexsym, mathtools}
\usepackage{amsthm}
\usepackage{ifthen}
\usepackage{hyperref}\hypersetup{colorlinks=true, unicode=true, linkcolor=[rgb]{0.10,0.05,0.67}, citecolor=[rgb]{0.10,0.05,0.67}, filecolor=[rgb]{0.10,0.05,0.67}, urlcolor=[rgb]{0.10,0.05,0.67}}
\usepackage{Shorthands}
\usepackage{cleveref}
\newcommand{\citep}[1]{\cite{#1}}
\usepackage{stackengine}

\newcommand{\fnorm}[1]{\ensuremath{\left\| #1 \right\|_\mathrm{F}}}
\newcommand{\sfnorm}[1]{\ensuremath{\| #1 \|_\mathrm{F}}}
\newcommand{\bfnorm}[1]{\ensuremath{\big\| #1 \big\|_\mathrm{F}}}
\newcommand{\bnorm}[1]{\ensuremath{\big\| #1 \big\|}}

\usepackage[numbers, sort&compress]{natbib}
\allowdisplaybreaks
\usepackage{booktabs}

\title{\LARGE \bf
Optimistic Online LQR via Intrinsic Rewards
}
\usepackage{subcaption}

\author{
Marcell Bartos$^\star$, Bruce D. Lee$^\star$, Lenart Treven, Andreas Krause, Florian Dörfler, Melanie N. Zeilinger
\thanks{$^\star$ Equal contribution (shared first authorship). All authors are with ETH Zürich, Zürich, Switzerland. Emails : \tt\small mbartos@ethz.ch, bruce.lee@ai.ethz.ch\normalfont\small. This work was supported as a part of NCCR Automation, a National Centre of Competence in Research, funded by the Swiss National Science Foundation (grant number 51NF40\textunderscore225155). This work is also partially supported by an ETH AI Center Fellowship to Bruce Lee.}%
}

\begin{document}

\twocolumn

\maketitle
\thispagestyle{empty}
\pagestyle{empty}

\begin{abstract}
Optimism in the face of uncertainty is a popular approach to balance exploration and exploitation in reinforcement learning. Here, we consider the online linear quadratic regulator (LQR) problem, i.e., to learn the LQR corresponding to an unknown linear dynamical system by adapting the control policy online based on closed-loop data collected during operation. In this work, we propose Intrinsic Rewards LQR (IR-LQR), an optimistic online LQR algorithm that applies the idea of intrinsic rewards originating from reinforcement learning and the concept of variance regularization to promote uncertainty-driven exploration. IR-LQR typically retains the structure of a standard LQR synthesis problem by only modifying the cost function, resulting in an intuitively pleasing, simple, computationally cheap, and efficient algorithm. This is in contrast to existing optimistic online LQR formulations that rely on more complicated iterative search algorithms or solve computationally demanding optimization problems. We show that IR-LQR achieves the optimal worst-case regret rate of $\sqrt{T}$, and compare it to various state-of-the-art online LQR algorithms via numerical experiments carried out on an aircraft pitch angle control and an unmanned aerial vehicle example.
\end{abstract}
\section{Introduction} 

Reinforcement learning (RL) has proven effective in a variety of applications from game playing \citep{mnih2013playing,mnih2015human,hafner2025training} to robotics \citep{levine2016end,intelligence2025pi}. Despite its success, it remains challenging to apply reliably in the real world due to sample inefficiency \citep{sutton1998reinforcement,duan2016benchmarking,henderson2018deep}. One of the critical bottlenecks in designing sample efficient RL algorithms is to efficiently navigate the exploration-exploitation tradeoff. 

Learning the linear quadratic regulator (LQR) through online interactions with an unknown system has emerged as a standard benchmark to study the efficacy of particular exploration schemes in continuous state and action spaces. Prior investigations have focused on strategies such as optimism in the face of uncertainty, and naive exploration \citep{abbasi2011regret, simchowitz2020naive}. Such approaches either sacrifice the computationally efficiency of LQR synthesis, or explore in an undirected manner. We study the use of uncertainty-based exploration bonuses for which the synthesis problem typically remains an LQR problem and show that this strategy achieves the optimal $\sqrt{T}$ regret rate.

Namely, the synthesis problem lying at the heart of the proposed algorithm can be informally written as
\begin{align*}
    K_t = \argmin_K J^\mathrm{LQR}(K,\hat{\Theta}_t,\tilde{Q}_t)
\end{align*}
where the informal notation $J^\mathrm{LQR}(K,\hat{\Theta}_t,\tilde{Q}_t)$ represents the LQR cost associated with controller $K$ under the current model estimate $\hat{\Theta}_t$ with a modified weight matrix $\tilde{Q}_t$.

\subsection{Related Work}

\paragraph{Exploration Strategies}
\looseness-1
Numerous effective strategies exist for navigating the exploration-exploitation tradeoff. One of the most celebrated is the principle of \emph{optimism in the face of uncertainty} (OFU), in which at every decision, the learner proposes a policy that maximizes the highest plausible return under uncertainty \citep{curi2020efficient, treven2023efficient}. While such approaches are statistically efficient, they tend to be computationally challenging, motivating strategies such as Thompson sampling \citep{chowdhury2017kernelized,russo2018tutorial} and $\varepsilon$-greedy approaches \citep{kearns2002near,mnih2015human,van2016deep}. Alternatively, the reward can be appended with exploration bonuses that encourage visiting novel regions in the state and action space \citep{strehl2008analysis, bellemare2016unifying, sukhija2025optimism}. These alternatives to optimism-based approaches still present computational challenges when extended to continuous state and action spaces, and analyses typically rely on oracle access to control synthesis algorithms over nonlinear dynamics and complex rewards. In this paper, we restrict our attention to linear dynamical systems and quadratic stage costs, and show that exploration bonuses yield simple and computationally efficient synthesis problems. 

\paragraph{Learning the LQR} 
\looseness-1
Learning the linear quadratic regulator has emerged as a standard benchmark for studying the efficiency of reinforcement learning applied to continuous state and action spaces. It has been shown that a learner interacting with an unknown linear system for $T$ time steps incurs at least $\Omega(\sqrt{T})$ regret \citep{simchowitz2020naive, ziemann2024regret}. \citet{abbasi2011regret} present an algorithm that achieves $\tilde{\calO}(\sqrt{T})$ regret using a computationally intractable optimism procedure. \citet{cohen2019learning} and \citet{abeille2020efficient} derive tractable optimism-based approaches based on semidefinite program reformulations and Lagrangian relaxations that maintain $\tilde{\calO}(\sqrt{T})$ regret. Simpler approaches based on Thompson sampling \citep{abeille2017thompson, kargin2022thompson} and naive exploration \citep{mania2019certainty, simchowitz2020naive}, which synthesize policies based on the principle of certainty equivalence and inject appropriately scheduled probing noise, can similarly achieve sublinear regret. While such simple schemes are effective, they leverage undirected exploration. Naive exploration is combined with the OFU principle in~\cite{lale2022reinforcement}, resulting in a statistically efficient but computationally expensive algorithm.  We analyze an approach based on exploration bonuses for learning the LQR that maintains the simplicity and efficiency of unmodified LQR  synthesis procedure, while promoting uncertainty-driven exploration.

\paragraph{Dual and Adaptive Control} 
The problem of dual control, in which an agent must interact with an uncertain system to simultaneously reduce uncertainty about the system and exploit the information that it has about the system, is introduced by \citet{feldbaum1960dual1, feldbaum1960dual2}. Determining the optimal dual controller is intractable in general. However, approximate schemes that estimate the system parameters via least squares and synthesize a controller via certainty equivalence have been shown to converge asymptotically to the optimal controller for the self-tuning regulator problem introduced by \cite{astrom1973self}. Subsequent work proposes a certainty equivalence-based algorithm that achieves asymptotically optimal regret for the self-tuning regulator problem \citep{lai1987asymptotically}. The self-tuning regulator problem studies minimum variance control of a plant modeled by ARX dynamics. The LQR variant of the self-tuning regulator problem possesses a more challenging exploration-exploitation tradeoff than minimum variance control, as playing a near-optimal controller does not provide sufficient information to identify the dynamics parameters for synthesis of the optimal controller~\citep{polderman1986necessity}.
Recent works in the field of adaptive control addressing the LQR problem with unknown dynamics include~\cite{iannelli2020multiobjective, zhao2025policy, bartos2025stability}.
In this work, we consider a dual exploration bonus that can also be understood as a type of variance regularization~\cite{pillonetto2022regularized, chiuso2023harnessing, zhao2025regularization}.

\subsection{Contribution}

We present an algorithm for learning the LQR through online interaction. The algorithm uses uncertainty-informed exploration bonuses to modify the control synthesis objective, thereby directing exploration to regions of the state and action space where the dynamics are most uncertain. The exploration bonus is chosen such that the synthesis problem remains an LQR problem, making it computationally efficient. By demonstrating that this synthesis objective is a lower bound on the cost achieved by the optimal controller (and therefore conforms to the OFU principle), we derive a high probability regret bound that scales with the mini-max optimal $\sqrt{T}$ rate, where $T$ is the number of interactions with the environment (i.e. time steps). The algorithm is evaluated on an airplane pitch control and an unmanned aerial vehicle example, and compared with alternatives such as naive exploration and other optimistic synthesis procedures. The experiments demonstrate the improved sample efficiency and computational efficiency of the approach.

\textbf{Outline: } The problem we aim to address is detailed in Section~\ref{sec:problem_formulation}. Section~\ref{sec:method} then presents the proposed method, and we analyze its performance in the form of a high-probability regret bound in Section~\ref{sec:analysis}. The results of the numerical experiments are presented in Section~\ref{sec:numerical}, which is followed by the conclusion in Section~\ref{sec:conclusion}. The proofs of the theoretical results can be found in the Appendix.

\textbf{Notation: } For vector $x$, we denote its Euclidean norm by $\snorm{x}$. We use $\succ$ ($\succeq$) to indicate positive (semi)definiteness. For matrix $A$, we use $\norm{A}_2$ and $\fnorm{A}$ to denote its spectral and Frobenius norm respectively, and $\norm{A}_P$ represents the weighted Frobenius norm $\fnorm{AP^{1/2}}$ for $P \succ 0$.
Consider the generalized discrete algebraic Riccati equation: $P = A^\top P A - (A^\top P B + N)(B^\top P B + R)^{-1} (B^\top P A + N^\top) + Q$. The unique positive definite solutions of the generalized discrete algebraic Riccati equation and the discrete Lyapunov equation are denoted by $\dare(A,B,Q,R, N)$ and $\dlyap(A,Q)$, respectively. The block diagonal matrix with blocks $Q$ and $R$ in its diagonal is denoted by $\mathrm{diag}(Q,R)$.  We use $\tilde{\calO}$ to hide logarithmic factors.
\section{Problem Formulation} \label{sec:problem_formulation}

We consider an unknown linear time invariant dynamical system which evolves according to
\begin{equation}
\label{eq: dynamics}
\begin{aligned}
    x_{t+1} &= A_\star x_t + B_{\star} u_t + w_t, t \geq 0
\end{aligned}
\end{equation}
with state $x_t \in \R^{\dx}$, input $u_t \in \R^{\du}$, process noise $w_t \overset{\iid}{\sim} \calN(0,\sigma_\mathsf{w}^2I)$\footnote{To simplify notation we restrict our attention to this simplified setting, but extensions to general $\Sigma_\mathsf{w} \succ 0$ noise covariances or even sub-Gaussian distributions are straightforward.} for $\sigma_\mathsf{w} > 0$, and initial condition $x_0 = 0$. We assume that we can directly measure the state $x_t$ and that the pair of unknown matrices $(A_{\star}, B_{\star})$ is stabilizable. We will use the shorthand notation $\Theta_\star := \bmat{A_\star & B_\star}$.

We consider the problem of learning to control \eqref{eq: dynamics} through online interactions. Let the control objective be defined by cost matrices $Q \succ 0$ and $R \succ 0$. The learner interacts with the system for $T$ time steps, and incurs a cumulative cost
\begin{align*}
    C_T:=\sum_{t=0}^{T-1}\ell(x_t, u_t), \quad \ell(x_t, u_t) := x_t^\top Q x_t + u_t^\top R u_t. 
\end{align*}
The performance of the learner is compared to the cost that could be achieved on average if the agent knew the system in advance. To this end, we define the ergodic LQR cost as
\begin{equation}
\begin{aligned}
    \label{eq: lqr objective}
    J(K, \Theta) := \limsup_{T
    \to\infty} \frac{1}{T} \E_{\Theta}^K \brac{\sum_{t=0}^{T-1} \ell(x_t, u_t)},
\end{aligned}
\end{equation}
where expectation is taken with respect to the process noise, the subscript denotes that the state evolves according to~\eqref{eq: dynamics} with $\Theta = \bmat{A & B}$ taking the place of $\Theta_\star$, and the superscript denotes that the actions are selected according to the controller $K$ as $u_t = K x_t$. We let $K_{\star} = \argmin_{K \in \R^{\du \times \dx}} J(K, \Theta_ \star)$ be the optimal controller.

The \emph{regret} of the learner is defined in the literature~\cite{abbasi2011regret} as 
\begin{align*}
    \calR_T := C_T - T J(K_{\star}, \Theta_\star). 
\end{align*}
A learner that approaches optimal behavior through the course of its interactions incurs sublinear regret. 

To show sublinear regret, we focus our attention on learning optimal behavior starting from a good initial guess. To this end, we introduce an initial set of parameters centered around the nominal estimate $\hat{\Theta}_0 := \bmat{\hat A_0 & \hat B_0}$ of radius $1/\lambda$:
\begin{align} \label{eq:calC0_def}
    \calC_0 := \curly{\Theta \in \mathbb{R}^{\dx \times (\dx + \du)}: \bfnorm{\Theta - \hat{\Theta}_0}\leq \frac{1}{\lambda}}.
\end{align}
We assume that this initial set contains the true parameter.
\begin{assumption} \label{asm:C0}
    We suppose that $\Theta_\star \in \calC_0$. We additionally suppose that all $\Theta \in \calC_0$ are stabilizable. 
\end{assumption}
Note that this assumption is not inherently restrictive in the sense that $\lambda$ could be arbitrarily small. However, algorithmic certificates will require a large enough $\lambda$. Without prior knowledge of this set, the learner could use warm-up schemes that achieve a coarse model of the system and an initial stabilizing controller,  as in~\cite{faradonbeh2018finite}.

Consider a more general version of~\eqref{eq: lqr objective}, where the positive definite cost to be minimized also contains cross terms: $\ell(x,u)=x^\top Q x + u^\top R u + 2x^\top Nu$. Then, for a known system, the minimizer is the LQR:
\begin{align}
    K(A,B, Q, R, N) &:= -(B^\top P B +R)^{-1} (B^\top P A + N^\top) \nonumber \\
    P &= \dare(A,B,Q,R, N).  \label{eq: gen lqr}
\end{align}
Our method will rely on the solution to such generalized LQR problems. The closed-form expression for the LQR controller via the Riccati equation makes the synthesis problem for a known dynamical system computationally efficient, and scalable to high-dimensional systems. In the sequel, we show that the learner may achieve sublinear regret using the principle of optimism in the face of uncertainty by solving LQR synthesis problems with a modified cost.

\section{Proposed Method} \label{sec:method}

\begin{figure*}[h!]
    \centering
    \includegraphics[width=0.99\textwidth]{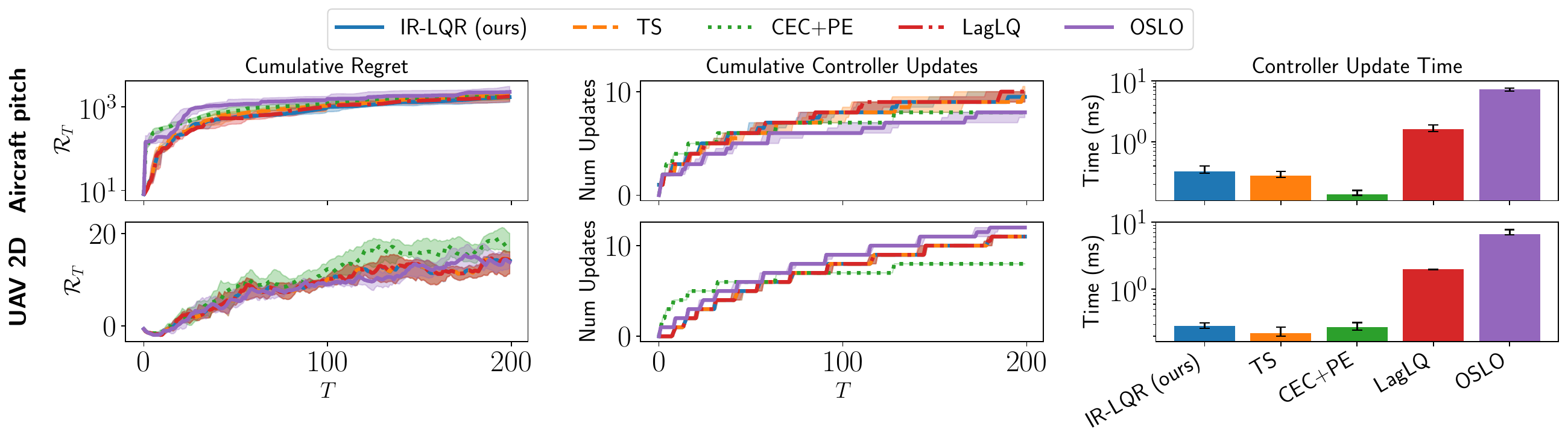}
    \caption{Comparison of online LQR methods on the aircraft pitch angle control (top) and unmanned aerial vehicle (bottom) examples, in terms of cumulative regret (left), number of controller updates (middle), and wall-clock time per controller update (right). The median is shown together with bootstrap confidence intervals (60\%) across runs. (40 samples, $T=200$)} \label{fig:plots}
    \vspace{-0.5cm}
\end{figure*}

We consider a learner that interacts with the system under a static state-feedback controller and periodically updates the controller gain using data collected from past interactions. To update the gain, the learner first estimates the unknown system parameters and then solves a control synthesis problem with a subtracted exploration bonus.

\subsection{Parameter Estimation}
To estimate the parameters of the system at an arbitrary time $t$, the learner uses the data $\calD_t = \curly{(x_k, u_k, x_{k+1})}_{k=0}^{t-1}$ to fit an estimate for $\Theta_\star$ using the projected regularized least squares estimator
\begin{align} \label{eq:pRLS}
    \hat{\Theta}_t = \Pi_{\calC_0} \left[ \left(\sum_{k=1}^t x_k\bmat{x_{k-1} \\ u_{k-1}}^\top + \lambda \hat{\Theta}_0\right) V^{-1}_t \right],
\end{align}
where $\Pi_{\calC_0}(\Theta):= \argmin_{\Theta' \in \calC_0} \norm{\Theta-\Theta'}_{V_t}$,
$V_t$ is the empirical covariance matrix
\begin{align} \label{eq:V_def}
    V_t :=\lambda I + \sum_{k=0}^{t-1} \bmat{x_k \\ u_k} \bmat{x_k \\ u_k}^\top,
\end{align}
and the regularizer $\lambda$ is the same variable that appears in the definition of $\calC_0$ in~\eqref{eq:calC0_def}. This estimator is commonly used in online LQR (cf.~\cite{abbasi2011regret, cohen2019learning, abeille2020efficient}) as it achieves the best fit to the data $\calD_t$ in a (regularized) least-squares sense. 
The regularization term is equivalent to placing an independent Gaussian prior on the entries of $\Theta$ centered at $\hat{\Theta}_0$, with variance $\sigma_\mathsf{w}^2/\lambda$ for each entry.
Assumption~\ref{asm:C0} is taken into account by projecting the estimate onto $\calC_0$.

The gap between this estimate and the ground truth is known to be bounded in terms of the empirical covariance matrix as follows.
\begin{lemma}(Adapted from~\cite[Prop.~1]{abeille2020efficient}.) \label{lem:id_success_event}
    For any $\delta_\mathrm{ID} \in (0,1)$, it holds with probability at least $1-\delta_\mathrm{ID}$ that
    \begin{align} \label{eq:id_success_event}
        \bnorm{\hat{\Theta}_t - \Theta_\star}_{V_t} \leq \beta_t(\delta_\mathrm{ID}),
    \end{align}
    $\forall t \geq 0$, where
    \begin{align} \label{eq:beta_def}
        \beta_t(\delta_\mathrm{ID}) :=  \sigma_\mathsf{w}\sqrt{2\dx\mathrm{log}\left( \frac{\dx\mathrm{det}(V_t)^{1/2}}{\delta_\mathrm{ID} \mathrm{det}(\lambda I)^{1/2}}\right)}+\frac{1}{\sqrt{\lambda}}.
    \end{align}
    Additionally, it also holds $\forall t \geq 0$ that
    \begin{align} \label{eq:lam_bound_on_error}
        \bfnorm{\hat{\Theta}_t - \Theta_\star} \leq \frac{2}{\lambda}.
    \end{align}
\end{lemma}

\subsection{Controller Synthesis}
\looseness-1
The learner uses the estimate $\hat{\Theta}_t$  to synthesize a controller to apply in subsequent time steps. It is well known that if the learner synthesizes and deploys a certainty equivalent policy greedily as $u_t = K(\hat A_t, \hat B_t, Q, R, 0) x_t$, it can get stuck in local minima in which the estimates of $\hat A_t$ and $\hat B_t$ do not improve \citep{simchowitz2020naive}. The learner then fails to achieve sublinear regret. This motivates approaches such as dithering that inject random noise at the actuation channel, as analyzed in \citep{simchowitz2020naive} in the online LQR setting; and more broadly it touches upon the concept of persistency of excitation in adaptive control~\cite{narendra2012stable}. Such exploration is undirected and continuously excites parameter directions also where the optimal controller already provides exploration. Optimistic approaches combat this inefficiency by using the uncertainty in the parameter estimates informed by \Cref{lem:id_success_event} to guide exploration more naturally. Unfortunately, they are typically computationally challenging. We therefore propose to use the estimate $\hat{\Theta}_t$ to synthesize an LQR controller with an exploration bonus. 

To this end, we define a modified control objective 
\begin{align}
    \tilde{J}&(K, \Theta, \calV) := \label{eq:J_tilde_def}\\
    &\limsup_{T
    \to\infty} \frac{1}{T} \E_{\Theta}^K \brac{\sum_{t=0}^{T-1} \bmat{x_t \\ u_t}^\top \left(\bmat{Q &  \\ & R} - \calV \right)\bmat{x_t \\ u_t}}, \nonumber
\end{align}
which can be understood as a generalized LQR problem where the cost has been modified by the matrix $\calV$. We will refer to the matrix $\calV$ as the exploration bonus. We consider an exploration bonus of the form 
\begin{align}
    \label{eq: variance bonus}
    \calV_t = g_t V_{t}^{-1},
\end{align} where
\begin{align}
    L_J &:= \frac{C^2(\fnorm{Q}+\bar{K}^2\fnorm{R})}{1-\rho^2},\label{eq:L_J} \\
    g_t &:= L_J(\beta_t^2 + 2\bar{\Theta}\beta_t \snorm{V_t}^{1/2}), \label{eq:gt} \\
    \bar{\Theta} &:= \max_{\Theta \in \calC_0}\sfnorm{\Theta},\quad \bar{K} := \max_{K \in \calK} \sfnorm{K}.
\end{align}
Such a bonus corresponds to minimizing a variance-regularized LQR cost function, i.e., the original cost is modified such that the state and input penalties are reduced along directions of the state and input space where past data provides less information (i.e., where the uncertainty encoded by $V_t^{-1}$ is large).
In other words, this cost modification promotes exploration, commonly referred to as an \emph{intrinsic reward} in reinforcement learning~\cite{bellemare2016unifying, sukhija2025optimism}. We will show that this bonus is \emph{optimistic}, i.e. that for any controller $K$, $\tilde{J}(K, \hat{\Theta}_t, \calV_t) \leq J(K, \Theta_\star)$ with high probability. 

The bonus $\calV_t$ may grow fast enough that the problem $\min_{K} J(K, \hat \Theta_t, \calV_t)$ becomes unbounded below. To account for this issue, we can define an alternative exploration bonus based on the prior Assumption~\ref{asm:C0}, which also leads to an optimistic objective. In particular, let 
\begin{align}
    c(\lambda) &:= 4L_J(\bar{\Theta}/\lambda + 1/\lambda^2). \label{eq:c(lam)}
\end{align}
We show that the bonus $\calV \gets c(\lambda) I$ is also optimsitic, i.e. it holds that $J(K, \hat \Theta_t, c(\lambda) I) \leq J(K, \Theta_\star)$. We therefore consider the synthesis problem defined by the maximum of these two objectives: 
\begin{align}
   \label{eq: optimistic synthesis}
   K_t = \argmin_K \max\{\tilde{J}(K,\hat{\Theta}_t, \calV_t),\tilde{J}(K, \hat{\Theta}_t, c(\lambda)I)\}.
\end{align}
This maximum objective inherits the optimism of the individual objectives. Furthermore, as long as 
\begin{align}
    \label{eq:lambdamin}
    c(\lambda) < \lambda_{\min}(\mathrm{diag}(Q,R)),
\end{align}
this objective remains bounded below. Whenever $\norm{\calV_t}\leq c(\lambda)$, the solution to the above problem is equivalent to
\begin{align} \label{eq:K_primary}
    K_t = \argmin_K \tilde{J}(K,\hat{\Theta}_t, \calV_t),
\end{align}
hence reducing to a standard LQR problem. When this condition is violated, the problem \eqref{eq: optimistic synthesis} can be solved via an efficient fallback procedure, for instance via an exact semidefinite program (SDP) reformulation detailed in Appendix~\ref{sec:app:SDP}. Our numerical experiments show that this fallback procedure is never needed in our examples.

We now proceed to show that the objective in \eqref{eq: optimistic synthesis} is optimistic. First, we state that a good enough initial estimate (i.e. small enough $\calC_0$) results in a set of controllers $\calK$ such that every element of $\calK$ stabilizes every element of $\calC_0$. Define the following certainty-equivalent controller:
\begin{align} \label{eq:Kstab_def}
        K_\mathrm{stab} = K(\hat{A}_0,\hat{B}_0,Q,R,0).
    \end{align}

\begin{assumption} \label{asm:Kstab}
    $K_\mathrm{stab}$ stabilizes $\Theta_\star$.
\end{assumption}

\begin{lemma} \label{lem:every_theta_stable}
    There exist $r_K>0$, $\lambda_\star<\infty$, $C\ge 1$, and $\rho\in(0,1)$ such that, for
    \begin{align} \label{eq:calK_def}
        \calK := \curly{K\in\mathbb R^{d_u\times d_x} : \fnorm{K-K_\mathrm{stab}} \leq r_K},    
    \end{align}
    it holds that $\forall \lambda \geq \lambda_\star$, $\forall \Theta=[A\ B]\in\mathcal C_0$, $\forall K\in\mathcal K$, and $\forall k\ge 0$,
    \begin{align*}
        \fnorm{(A + BK)^k} \leq C\rho^k.
    \end{align*}
\end{lemma}
Next we show that the aforementioned optimistic synthesis procedure results in controllers that lie in the set $\calK$.
\begin{lemma} \label{lem:K_in_calK}
Consider $K_\mathrm{stab}$ and $\calK$ as defined in~\eqref{eq:Kstab_def} and~\eqref{eq:calK_def} for some $r_K>0$.
Then there exists $\lambda_\star<\infty$ such that $\forall \lambda\geq \lambda_\star$ and $\forall t \geq 0$, controllers $K_\star$ as well as $K_t$ given by \eqref{eq: optimistic synthesis} belong to $\mathcal K$.
\end{lemma}
Finally, we can state the aforementioned lower bound result.

\begin{theorem} \label{thm:optimstic_is_lower_bound}
    Suppose that~\eqref{eq:id_success_event} holds for some time step $t$ and that $\lambda$ is large enough such that~\eqref{eq:lambdamin} and Lemmas~\ref{lem:every_theta_stable}-\ref{lem:K_in_calK} apply. Then $\forall K \in \calK$, it holds that
    \begin{align} \label{eq:Jbar_optimistic}
        \max\{\tilde{J}(K,\hat{\Theta}_t, \calV_t),\tilde{J}(K, \hat{\Theta}_t, c(\lambda)I)\} \leq J(K, \Theta_\star).
    \end{align}
\end{theorem}
Theorem~\ref{thm:optimstic_is_lower_bound} provides an upper bound of the performance of the optimistic controller in terms of the optimal controller, enabling us to derive sublinear regret in the following. It uses the intermediate result that the difference in the performance of a given controller under the true and estimated system is bounded by the exploration bonus, which in turn is made possible by the appropriate design of the exploration bonus. In particular,
$g_t$ in~\eqref{eq:gt} scales the bonus such the the aforementioned upper bound holds, and it also ensures that it decays with the appropriate rate of $\tilde{\calO}(1/\sqrt{t})$.
The proofs of Lemmas~\ref{lem:every_theta_stable}-\ref{lem:K_in_calK} and Theorem~\ref{thm:optimstic_is_lower_bound} can be found in Appendix~\ref{sec:app:normlemma}, \ref{sec:app:calKlemma}, and~\ref{sec:app:optim}.

\subsection{Algorithm and Discussion}

Building on the observation that this synthesis procedure is optimistic, we draw inspiration from existing methods following the OFU principle~\cite{abbasi2011regret, cohen2019learning, abeille2020efficient} to devise an algorithm with sublinear regret guarantees. The proposed algorithm (IR-LQR) 
is presented in Algorithm~\ref{alg:IR-LQR}.

\begin{algorithm}
\caption{Online LQR with Intrinsic Rewards (IR-LQR)} \label{alg:IR-LQR}
\begin{algorithmic}[1]
\State \textbf{Initialization:} set $t = 0$, $\tau = 0$, $m=0$, synthesize $K_0$ according to~\eqref{eq: optimistic synthesis} using initial estimates $\hat{A}_0$, $\hat{B}_0$, and apply $u_0=K_0x_0$
\For{$t = 1, \dots, T-1$}
    \State observe $x_t$
    \State compute $V_t$ \eqref{eq:V_def}
    \If{$\det(V_t) > 2\det(V_\tau)$ \textbf{ and } $t - \tau \geq \underline{\tau}(m, t)$ } \label{algstep:info_gain_criterion}
        \State set $\tau = t$, $m \gets m+1$
        \State compute $\hat{\Theta}_t$ \eqref{eq:pRLS} and $g_t$ \eqref{eq:gt}
        \State synthesize $K_t$ according to~\eqref{eq: optimistic synthesis}
    \Else
        \State set $K_t = K_{t-1}$
    \EndIf
    \State apply $u_t = K_tx_t$
\EndFor
\end{algorithmic}
\end{algorithm}

\looseness-1
Algorithm~\ref{alg:IR-LQR} is non-episodic in the sense that it is applied to a single continuous trajectory of the system without requiring restarts. However, as is customary in online LQR~\cite{abbasi2011regret, cohen2019learning, abeille2020efficient, simchowitz2020naive}, it operates in \emph{epochs}, during which the control policy is kept fixed. After observing the current state $x_t$ and updating the empirical covariance matrix $V_t$, IR-LQR checks whether sufficient new information has been collected since the start of the current epoch (denoted by time index $\tau$). In particular, the criterion is that the information gain~\citep{shannon1948mathematical} difference $I(\Theta; \calD_t) - I(\Theta; \calD_\tau)$ exceeds a constant threshold of $\dx/2$ bits, which can be equivalently expressed as $\det(V_t) > 2 \det(V_\tau)$ in our setting. When this criterion is met and the epoch length exceeds a lower bound $\underline{\tau}(m, t)$, a new controller is synthesized and a new epoch begins. The minimum epoch length is set long enough to avoid switching too frequently, and it is defined in~\eqref{eq:tau_underline_def} in the Appendix. Such a combination of an information gain criterion and a minimum epoch length for online LQR has also been used in~\cite{lale2022reinforcement}.
\section{Regret Analysis} \label{sec:analysis}

We now proceed by stating the main result of the paper, which is a high probability bound on the regret of the proposed IR-LQR algorithm (Algorithm~\ref{alg:IR-LQR}).

\begin{theorem} \label{thm:IR-LQR_regret}
    Let $\delta \in (0,1)$. Choose $\lambda$ (depending on $\delta$ but independently from $T$) large enough such that~\eqref{eq:lambdamin} and Lemmas~\ref{lem:every_theta_stable}-\ref{lem:K_in_calK} apply, and suppose that Assumption~\ref{asm:C0} holds. Then, with probability at least $1-\delta$, IR-LQR achieves the cumulative regret $\calR_T^\text{IR-LQR} = \Tilde{\calO}(\sqrt{T}).$
\end{theorem}
\begin{proof}
    The proof can be found in Appendix~\ref{sec:app:main}.
\end{proof}

Prior work has shown that $\calR_T = \tilde{\calO}(\sqrt{T})$ is the best attainable rate in general for the online LQR problem~\cite{simchowitz2020naive, ziemann2024regret}. This lower bound has  been achieved up to problem constants and polylogarithmic terms by a variety of existing methods including other optimism-based approaches~\cite{abbasi2011regret, cohen2019learning, abeille2020efficient}, as well as methods following certainty equivalent design combined with naive, undirected exploration~\cite{mania2019certainty, simchowitz2020naive}.

However, to the authors' best knowledge, this is the first work that shows that $\calR_T = \tilde{\calO}(\sqrt{T})$ can be achieved by using an exploration bonus which typically maintains the structure of the LQR problem. Most closely related is the method proposed in~\cite{cohen2019learning} that modifies the SDP formulation of the LQR, leading to a computationally harder synthesis problem, and more complicated analysis. Beyond the LQR problem, the regret of using similar intrinsic rewards for model-based RL has been shown to scale with $\sqrt{T}$ for nonlinear dynamics represented by Gaussian processes~\cite{sukhija2025optimism}; however, these results rely on a synthesis oracle. We show that in the linear setting, a modification of this results in synthesis problems that are  generalized LQR problems. A further advantage of IR-LQR compared to other optimistic online LQR methods is that the required accuracy of the initial estimate $\hat{\Theta}_0$ is independent of $T$, whereas the initial estimate for~\cite{abeille2020efficient} and~\cite{cohen2019learning} needs to be accurate within $1/\lambda = \calO(1/\log T)$ and $1/\lambda = \calO(1/\sqrt{T})$, respectively. A summary of this discussion can be found in Table~\ref{tab:methods_comparison}.

\begin{table}
\renewcommand{\arraystretch}{1.3}
\caption{Comparison of online LQR algorithms.}
\label{tab:methods_comparison}
\begin{center}
    \begin{tabular}{@{}lccc@{}}
    \toprule
     &  \shortstack{Synthesis \\ problem} & \shortstack{Noise \\ injection} &  \shortstack{Requirement \\ on $1/\lambda$} \\
    \midrule
    CEC+PE~\cite{simchowitz2020naive} & LQR & yes & $\calO(1)$\footnotemark\\
    OSLO~\cite{cohen2019learning} & SDP & no & $\calO(1/\sqrt{T})$ \\
    LagLQ~\cite{abeille2020efficient} & iterative search & no & $\calO(1/\log T)$ \\ 
    TS~\cite{abeille2017thompson} & LQR & yes\footnotemark & $\calO(1)$ \\
    IR-LQR (ours) & LQR\footnotemark & no & $\calO(1)$\\
    \bottomrule
\end{tabular}
\end{center}
\vspace{-0.5cm}
\end{table}
\footnotetext[2]{This method uses a warm-up phase rather than an initial estimate, but the regret analysis presented there could be modified to instead work with a prior that is independent of $T$.}
\footnotetext[3]{The method in~\cite{abeille2017thompson} does not inject noise, but their theory only applies to scalar systems. The state-of-the-art theory of Thompson sampling in online LQR~\cite{kargin2022thompson} requires noise injection.}
\footnotetext[4]{Only true for the primary controller, while the fallback problem can be exactly reformulated as an SDP.}
\section{Numerical Experiments} \label{sec:numerical}

In this section, we compare IR-LQR with state-of-the-art online LQR algorithms on an aircraft pitch angle control example\footnote{The example was taken from ``Control Tutorials for Matlab and Simulink", University of Michigan \url{https://ctms.engin.umich.edu/CTMS/?example=AircraftPitch&section=SystemModeling}}, and a UAV control problem taken from~\cite{lale2022reinforcement}.
The code to reproduce the results can be found online, along with several additional benchmark problems.\footnote{\url{https://github.com/lenarttreven/lqr_research} \\ The experiments are carried out using Python on a machine equipped with an Intel i7-12700H (2.30 GHz) processor with 32 GB of RAM.} The considered methods are: (1) Intrinsic Reward LQR (IR-LQR, Algorithm~\ref{alg:IR-LQR}), (2) Certainty Equivalence Control with Persistent Excitation (CEC+PE~\cite{simchowitz2020naive}), (3) Optimistic Semidefinite Programming for LQ Control (OSLO~\cite{cohen2019learning}), (4) Optimistic LQR with Lagrangian Relaxation (LagLQ~\cite{abeille2020efficient}), and (5) Thompson Sampling LQR (TS~\cite{abeille2017thompson}).

We carried out hyperparameter tuning for all of the algorithms. In particular, $g_t$ and $c(\lambda)$ depend on quantities that are not known in practice. Therefore, we set $c(\lambda) = 0.95\lambda_{\min}(\diag(Q,R))$ and $g_t = \gamma_1 +\gamma_2\|V_t\|^{1/2}$, where $\gamma_{1,2}$ were tuned via cross-validation. We sampled 40 random noise sequence realizations and ran the experiments with horizon $T=200$. The results are summarized in Figure~\ref{fig:plots}. 
The plots on the left show the cumulative regret suffered by the methods. The observed regret curves were consistent with sublinear growth, and their performances are comparable, although IR-LQR achieves the lowest median regret among them in both cases. We observed that IR-LQR is robust to the choice of hyperparameters in terms of success of the experiments and partially also with respect to the accumulated cost. Additionally, the fallback controller was never used in either of the two examples above, highlighting that it is primarily needed for the theoretical analysis.

We also compared the algorithms in terms of computational load. As expected, the methods that rely on solving more complicated synthesis problems (LagLQ, OSLO) are computationally significantly more expensive compared to those that rely on LQR synthesis problems (IR-LQR, TS, CEC+PE): while the number of controller updates are almost identical across them, the time it takes to carry out a single controller update is an order of magnitude slower for the former class of methods compared to the latter. Note that IR-LQR remained computationally efficient owing to the fact that the fallback option was never used.

In summary, as highlighted by the numerical example as well as Table~\ref{tab:methods_comparison}, IR-LQR successfully utilizes directed exploration to achieve state-of-the-art performance without compromising computational effort compared to undirected approaches.
\section{Conclusion} \label{sec:conclusion}

We proposed IR-LQR, an algorithm based on the optimism in the face of uncertainty principle to solve the online LQR problem. Unlike existing optimism-based methods, our method seeks to retain the structure of the original LQR synthesis problem, resulting in an intuitively pleasing, conceptually simple, computationally cheap, and typically efficient algorithm. We achieved this by leveraging the idea of intrinsic rewards from the field of reinforcement learning. Additionally, it does not require persistently exciting noise injection to learn the optimal controller. We proved that IR-LQR matches the optimal worst-case regret rate of $\sqrt{T}$, validated the theoretical findings via numerical experiments, and compared IR-LQR to several state-of-the-art online LQR algorithms.


\begingroup
\small
\bibliographystyle{IEEEtranN}
\bibliography{refs}
\endgroup

\section*{Appendix}

For the proofs, we introduce the following notation:
\begin{align}
    \bar{J}(K,\Theta,\calV)&:=\max \{\tilde{J}(K,\Theta, \calV),\tilde{J}(K, \Theta, c(\lambda)I)\} \nonumber\\
    &=J(K,\Theta) - \calB(K,\Theta,\calV), \nonumber\\
    \calB(K,\Theta,\calV) &:= \min \{ \trace(\calV\tilde{\Sigma}_\Theta^K),c(\lambda)\trace(\tilde{\Sigma}_\Theta^K)\},\label{eq:calB} \\
    \tilde{\Sigma}^K_\Theta := \bmat{I \\ K} &\Sigma^K_\Theta \bmat{I \\ K}^\top,\ \Sigma^K_\Theta := \dlyap(\Theta \bmat{I \\ K}, \sigma_\mathsf{w}^2I).\nonumber
\end{align}
Additionally, we remark that the LQR costs can be equivalently written as
\begin{align*}
    J(K, \Theta) &= \trace(\mathrm{diag}(Q,R)\tilde{\Sigma}^{K}_\Theta) \\
    &=\trace((Q + K^\top RK)\Sigma^K_\Theta),\\
    \tilde{J}(K, \Theta, \calV) &= \trace((\mathrm{diag}(Q,R) - \calV)\tilde{\Sigma}^{K}_\Theta).
\end{align*}

\subsection{SDP reformulation for the fallback controller} \label{sec:app:SDP}
Consider the SDP
\begin{subequations} \label{eq:SDP}
    \begin{align}
    \min_{\Sigma \succeq 0,s} \quad\quad &s \\
    \text{s.t.} \quad \Sxx &= \hat{\Theta}_t\Sigma \hat{\Theta}_t^\top + \sigma_\mathsf{w}^2I \label{eq:SDP:dynamics}\\
    s &\geq \trace((\mathrm{diag}(Q,R)-\calV_t)\Sigma) \label{eq:SDP:tr1} \\
    s &\geq \trace((\mathrm{diag}(Q,R)-c(\lambda)I)\Sigma) \label{eq:SDP:tr2}
    \end{align}
\end{subequations}
with
\begin{align*}
    \Sigma = \begin{bmatrix}
        \Sxx & \Sxu \\ \Sxu^\top & \Suu
    \end{bmatrix} \in \mathbb{R}^{(\dx+\du)\times(\dx+\du)}.
\end{align*}
Since any solution of~\eqref{eq:SDP} with $\rank(\Sigma) = \dx$ corresponds to a deterministic linear state feedback matrix, taking such a solution and defining
\begin{align*}
    K_t = \Sxu^\top \Sxx^{-1}
\end{align*}
results in a controller that solves the fallback synthesis problem $K_t = \argmin_K \bar{J}(K,\hat{\Theta}_t,\calV_t)$.

Next, we prove that an optimal solution of~\eqref{eq:SDP} with $\rank(\Sigma) = \dx$ always exists. Take an optimal $\Sigma$ and define $r:=\rank(\Sigma)$. From $\Sxx \succeq \sigma_\mathsf{w}^2I$, it holds that $r \geq \dx$. Suppose $r > \dx$. The goal is to transform such a solution to another solution with rank $\dx$. Factor $\Sigma = UU^\top$ with full column rank matrix $U \in \mathbb{R}^{(\dx+\du)\times r}$ and define the perturbation
\begin{align*}
    \Sigma(\epsilon) := U(I+\epsilon H)U^\top,
\end{align*}
where $H = H^\top \in \mathbb{R}^{r \times r}$. Consider the space $\calN$ of symmetric perturbations $H$ under which the dynamics constraint~\eqref{eq:SDP:dynamics} is satisfied. Then
\begin{align*}
    \dim(\calN) \geq \frac{r(r+1)}{2}-\frac{\dx(\dx+1)}{2} \geq \dx+1 \geq 2,
\end{align*}
where $\dx(\dx+1)/2$ is the maximum number of independent scalar equality constraints imposed by~\eqref{eq:SDP:dynamics}.

Define
\begin{align*}
    \phi_1(H) &:= \trace(U^\top (\mathrm{diag}(Q,R)-\calV_t)U H ), \\
    \phi_2(H) &:= \trace(U^\top (\mathrm{diag}(Q,R)-c(\lambda)I)U H ).
\end{align*}
Since $\dim(\calN) \geq 2$, $\exists H \in \calN, H\neq0$ such that $\phi_1(H) \leq 0$ and $\phi_2(H) \leq 0$. Therefore, for such a choice of $H$, $\Sigma(\epsilon)$ satisfies constraints~\eqref{eq:SDP:tr1}-\eqref{eq:SDP:tr2} $\forall \epsilon >0$.

Due to~\eqref{eq:lambdamin} and the fact that $U$ is full column rank, $U^\top (\mathrm{diag}(Q,R)-c(\lambda)I)U \succ 0$. This means that if $H \succeq 0$, then $\phi_2(H) > 0$, contradicting the fact that $H$ was chosen such that $\phi_2(H) \leq 0$. That is, the chosen $H$ must have at least one negative eigenvalue: $\lambda_{\min}(H) < 0$. Select $\varepsilon := -1/\lambda_{\min}(H) > 0$, for which $I+\varepsilon H$ is singular, and thus $\rank(\Sigma(\varepsilon)) < \rank(\Sigma)$. Furthermore, as $H$ was chosen such that $\Sigma(\varepsilon)$ satisfies every constraint in~\eqref{eq:SDP}, it is also an optimal solution of~\eqref{eq:SDP}, i.e., we successfully constructed a new solution with strictly smaller rank. This procedure can be repeated until we end up with an optimal solution that satisfies $\rank(\Sigma) = \dx$.

\subsection{Proof of Lemma~\ref{lem:every_theta_stable}} \label{sec:app:normlemma}
\begin{proof}
Let $M_0 := \hat A_0 + \hat B_0 K_{\mathrm{stab}}$. Since $M_0$ is Schur,
there exists $P \succ 0$ satisfying
\begin{align*}
    M_0^\top P M_0 - P = -I.
\end{align*}
For $M=M_0+\Delta$, it follows that
\begin{align*}
\begin{aligned}
M^\top P M-P
&=
-I+M_0^\top P\Delta+\Delta^\top P M_0+\Delta^\top P\Delta\\
&\preceq
\left(
-1
+2\|P\|_2\|M_0\|_2\|\Delta\|_2
+\|P\|_2\|\Delta\|_2^2
\right)I.
\end{aligned}
\end{align*}
Hence, there exists $\epsilon>0$ such that
$\|\Delta\|_2\leq\epsilon$ implies
\begin{align*}
    M^\top P M
    \preceq P-\frac{1}{2}I
    \preceq (1-\eta)P,
    \qquad
    \eta:=\frac{1}{2\|P\|_2}\in(0,1).
\end{align*}

For any $\Theta=[A\;B]\in\mathcal C_0$ and $K\in\mathcal K$,
\begin{align*}
&\|A+BK-M_0\|_2 \\
&\leq
\|A-\hat A_0\|_2
+\|(B-\hat B_0)K\|_2
+\|\hat B_0(K-K_{\mathrm{stab}})\|_2\\
&\leq
\frac{1+\sfnorm{K_{\mathrm{stab}}}+r_K}{\lambda}
+\|\hat B_0\|_2r_K .
\end{align*}
Choose $r_K>0$ sufficiently small and then $\lambda_\star<\infty$
sufficiently large such that the right-hand side is at most $\epsilon$
for all $\lambda\geq\lambda_\star$. Thus, for every
$\Theta=[A\;B]\in\mathcal C_0$ and $K\in\mathcal K$,
\begin{align*}
    (A+BK)^\top P(A+BK)\preceq(1-\eta)P.
\end{align*}
Consequently, for all $k\geq0$,
\begin{align*}
    \|(A+BK)^k\|_2
    \leq \sqrt{\kappa(P)}(1-\eta)^{k/2},
\end{align*}
and therefore
\begin{align*}
    \sfnorm{(A+BK)^k}
    \leq \sqrt{\dx\kappa(P)}(1-\eta)^{k/2}.
\end{align*}
The claim follows with
$C=\sqrt{\dx\kappa(P)}$ and $\rho=\sqrt{1-\eta}$.
\end{proof}

\subsection{Proof of Lemma~\ref{lem:K_in_calK}} \label{sec:app:calKlemma}
\emph{1) $K_\star$ and the primary controller~\eqref{eq:K_primary}:}
Since $\hat{\Theta}_{\tau_m}\in\calC_0$, we have
\begin{align*}
    \bfnorm{\hat{\Theta}_{\tau_m}-\hat\Theta_0} \leq 1/\lambda.
\end{align*}
Moreover, whenever~\eqref{eq:K_primary} is used, $\snorm{\calV_{\tau_m}} \leq c(\lambda)$.
Because $c(\lambda)=\calO(1/\lambda)$ and the stabilizing generalized LQR controller depends continuously on the problem data in a neighborhood of $(\hat{A}_0,\hat{B}_0,Q,R,0)$, there exist constants $L_\Theta,L_V$ such that
\begin{align*}
    \fnorm{K_{\tau_m}-K_\mathrm{stab}} &\leq L_\Theta \bfnorm{\hat{\Theta}_{\tau_m}-\hat{\Theta}_0} + L_V \norm{\calV_{\tau_m}}_2 \\
    &\leq L_\Theta/\lambda + L_V c(\lambda),
\end{align*}
see also~\cite[Lemma~3]{bartos2025stability}. Choosing $\lambda$ sufficiently large yields
$\fnorm{K_{\tau_m}-K_\mathrm{stab}} \leq r_K$, hence $K_{\tau_m}\in\calK$. $K_\star\in\calK$ can be proven analogously.

2) \emph{Fallback controller~\eqref{eq: optimistic synthesis}:}
Let
\[
K_\circ := \argmin_K \tilde J(K,\hat\Theta_{\tau_m},
c(\lambda)I).
\]
By the same continuity argument as above,
\begin{align*}
    \|K_\circ-K_{\rm stab}\|_{\rm F}
    &\leq L_\Theta
    \|\hat\Theta_{\tau_m}-\hat\Theta_0\|_{\rm F}
    +L_c c(\lambda)
    =O(1/\lambda).
\end{align*}
Thus, for sufficiently large $\lambda$, $K_\circ\in\mathcal K$. In
particular, Lemma III.2 implies
\begin{align*}
    \tr(\tilde\Sigma_{\hat\Theta_{\tau_m}}^{K_\circ})
    \leq
    \frac{(1+\bar K^2)\sigma_\mathsf{w}^2 C^2}{1-\rho^2}.
\end{align*}
Moreover, optimality of the fallback controller and
$\mathcal V_{\tau_m}\succeq0$ yield
\begin{align*}
    \tilde J(K_{\tau_m},\hat\Theta_{\tau_m},c(\lambda)I)
    &\leq
    \bar J(K_{\tau_m},\hat\Theta_{\tau_m},\mathcal V_{\tau_m})\\
    &\leq
    \bar J(K_\circ,\hat\Theta_{\tau_m},\mathcal V_{\tau_m})\\
    &\leq
    \tilde J(K_\circ,\hat\Theta_{\tau_m},c(\lambda)I)
    +c(\lambda)
    \tr(\tilde\Sigma_{\hat\Theta_{\tau_m}}^{K_\circ}).
\end{align*}
Since $c(\lambda)<\lambda_{\min}(\diag(Q,R))$, the modified
cost with exploration bonus $c(\lambda)I$ is positive definite.
Consequently, the above inequality also implies that
$K_{\tau_m}$ stabilizes $\hat\Theta_{\tau_m}$. Applying the
LQR performance-difference result
\cite[Lemma~10]{fazel2018global} to this modified LQR problem
therefore gives
\begin{align*}
    &\tilde J(K_{\tau_m},\hat\Theta_{\tau_m},c(\lambda)I)
    -\tilde J(K_\circ,\hat\Theta_{\tau_m},c(\lambda)I)\\
    &\qquad\geq
    \sigma_\mathsf{w}^2\bigl(\lambda_{\min}(R)-c(\lambda)\bigr)
    \|K_{\tau_m}-K_\circ\|_{\rm F}^2.
\end{align*}
Combining the previous inequalities, and choosing $\lambda$
sufficiently large such that
$c(\lambda)\leq\lambda_{\min}(R)/2$, yields
\begin{align*}
    \|K_{\tau_m}-K_\circ\|_{\rm F}^2
    &\leq
    \frac{c(\lambda)
    \tr(\tilde\Sigma_{\hat\Theta_{\tau_m}}^{K_\circ})}
    {\sigma_\mathsf{w}^2(\lambda_{\min}(R)-c(\lambda))}
    =O(c(\lambda))
    =O(1/\lambda).
\end{align*}
Hence $\|K_{\tau_m}-K_{\rm stab}\|_{\rm F}
    =O(1/\sqrt{\lambda})$, and
a sufficiently large $\lambda$  therefore yields
$\|K_{\tau_m}-K_{\rm stab}\|_{\rm F}\leq r_K$, i.e.,
$K_{\tau_m}\in\mathcal K$.
\hfill $\blacksquare$

\subsection{Proof of Theorem~\ref{thm:optimstic_is_lower_bound} and further auxiliary results} \label{sec:app:optim}
The proof relies on the following result on the perturbation of Lyapunov equations:
\begin{lemma} \label{lem:Lyapunov_perturbation}
    Let $\Sigma = \dlyap(A,Q)$ and $\Sigma' = \dlyap(A',Q)$ with $Q \succ 0$ and $A,A'$ such that $\exists C \geq 1$ and $\exists\rho \in (0,1)$ satisfying $\| A^k\| \leq C\rho^k$ and $\| A'^k\| \leq C\rho^k$ for all $k\geq0$. Let $\|.\|$ denote any given matrix norm. Then
    \begin{align*}
        \| &\Sigma - \Sigma' \| \leq \frac{C^2}{1-\rho^2} \\
        &\times (2 \| A \Sigma' (A-A')^\top\| +  \| (A-A') \Sigma' (A-A')^\top\|).
    \end{align*}
\end{lemma}
\begin{proof}
    \begin{align*}
        \Sigma' - \Sigma &= (A + A' - A) \Sigma' (A + A' - A)^\top - A \Sigma A^\top \\
        &= \dlyap(A, \tilde{Q}),
    \end{align*}
    where
    \begin{align*}
        \tilde{Q} = (A'-A)\Sigma' A^\top &+ A \Sigma' (A' - A)^\top \\
        &+ (A'-A)\Sigma' (A'-A)^\top,
    \end{align*}
    and thus
    \begin{align*}
        \|\Sigma - \Sigma'\| &= \| \sum_{k=0}^\infty A^k \tilde{Q} (A^\top)^k \| \\
        &\leq \| \tilde{Q} \| \sum_{k=0}^\infty \| A^k\|^2 \leq \frac{C^2}{1-\rho^2} \| \tilde{Q}\|,
    \end{align*}
    where $\| \tilde{Q} \| \leq 2 \| A \Sigma' (A-A')^\top\| +  \| (A-A') \Sigma' (A-A')^\top\|$.
\end{proof}

Lemma~\ref{lem:Lyapunov_perturbation}, enables us to bound the difference in the cost of a given controller on the true system and the estimated system by the exploration bonus, which is formally stated in the following lemma:
\begin{lemma} \label{lem:exploration_bounds_cost_diff}
    Suppose that~\eqref{eq:id_success_event} holds for some time step $t$.
    Then, for any controller $K \in \calK$, it holds that
    \begin{align} \label{eq:proof:optimism_theta_hat}
        | J(K, \hat{\Theta}_t) - J(K, \Theta_\star) | \leq \calB(K,\hat{\Theta}_t,\calV_t).
    \end{align}
\end{lemma}
\begin{proof}
    It immediately follows from Lemma~\ref{lem:Lyapunov_perturbation} that
    \begin{align*}
        &\bfnorm{\Sigma^K_{\Theta_\star} -  \Sigma^K_{\hat{\Theta}_t}} \\ &\leq \frac{C^2}{1-\rho^2}\Big( 2 \fnorm{\Theta_\star}  \underbrace{\bfnorm{\tilde{\Sigma}^K_{\hat{\Theta}_t} (\Theta^\star-\hat{\Theta}_t)^\top}}_{\diamondsuit} \\
        &+ \underbrace{\bfnorm{ (\Theta^\star-\hat{\Theta}_t) \tilde{\Sigma}^K_{\hat{\Theta}_t} (\Theta^\star-\hat{\Theta}_t)^\top}}_{\spadesuit}\Big).
    \end{align*}
    Here $\diamondsuit$ and $\spadesuit$ can be bounded in two different ways. On one hand, using~\eqref{eq:lam_bound_on_error},
    \begin{align*}
        \diamondsuit &= \sqrt{\mathrm{tr}\left(  (\Theta^\star-\hat{\Theta}_t) \tilde{\Sigma}^K_{\hat{\Theta}_t}\tilde{\Sigma}^K_{\hat{\Theta}_t} (\Theta^\star-\hat{\Theta}_t)^\top\right)} \\ &\leq \frac{2}{\lambda}\sqrt{\mathrm{tr}(\tilde{\Sigma}^K_{\hat{\Theta}_t} \tilde{\Sigma}^K_{\hat{\Theta}_t})} \leq \frac{2}{\lambda}\mathrm{tr}(\tilde{\Sigma}^K_{\hat{\Theta}_t}),
    \end{align*}
    and
    \begin{align*}
        \spadesuit &=\sqrt{\mathrm{tr}\left( ( (\Theta^\star-\hat{\Theta}_t) \tilde{\Sigma}^K_{\hat{\Theta}_t} (\Theta^\star-\hat{\Theta}_t)^\top)^2\right)} \\
        &\leq \mathrm{tr}\left(  (\Theta^\star-\hat{\Theta}_t) \tilde{\Sigma}^K_{\hat{\Theta}_t} (\Theta^\star-\hat{\Theta}_t)^\top \right) \leq \frac{4}{\lambda^2}\mathrm{tr}(\tilde{\Sigma}^K_{\hat{\Theta}_t}) 
    \end{align*}
    yielding
    \begin{align}
        |J&(K, \hat{\Theta}_t) -J(K, \Theta_\star) | \label{eq:J_diff_bound_1}\\ &\leq \bfnorm{Q + K^\top R K} \bfnorm{\Sigma^K_{\Theta_\star} -  \Sigma^K_{\hat{\Theta}_t}} 
        \leq c(\lambda)\trace(\tilde{\Sigma}^K_{\hat{\Theta}_t}). \nonumber
    \end{align}
    On the other hand,~\eqref{eq:id_success_event} results in
    \begin{align*}
        \diamondsuit & \leq \beta_t \sqrt{\trace( \tilde{\Sigma}^K_{\hat{\Theta}_t} V_t^{-1} \tilde{\Sigma}^K_{\hat{\Theta}_t})} \\
        &\leq \beta_t \bnorm{V^{1/2}}_2\sqrt{\trace(V_t^{-1} \tilde{\Sigma}^K_{\hat{\Theta}_t} V_t^{-1} \tilde{\Sigma}^K_{\hat{\Theta}_t})} \\
        &\leq \beta_t \bnorm{ V^{1/2}}_2 \trace(V_t^{-1} \tilde{\Sigma}^K_{\hat{\Theta}_t}),
    \end{align*}
    and
    \begin{align*}
        \spadesuit \leq \beta_t^2 \trace(V_t^{-1} \tilde{\Sigma}^K_{\hat{\Theta}_t}),
    \end{align*}
    i.e.,
     \begin{align}
        |J&(K, \hat{\Theta}_t) -J(K, \Theta_\star) | \label{eq:J_diff_bound_2}\\ &\leq \bfnorm{Q + K^\top R K} \bfnorm{\Sigma^K_{\Theta_\star} -  \Sigma^K_{\hat{\Theta}_t}} 
        \leq g_t\trace(V_t^{-1}\tilde{\Sigma}^K_{\hat{\Theta}_t}). \nonumber
    \end{align}
    Combining~\eqref{eq:calB}, \eqref{eq:J_diff_bound_1}, and~\eqref{eq:J_diff_bound_2} proves~\eqref{eq:proof:optimism_theta_hat}.
\end{proof}
Note that this result immediately implies Theorem~\ref{thm:optimstic_is_lower_bound}.

\subsection{Success events} \label{sec:app:events}

Before proceeding with the proof of Theorem~\ref{thm:IR-LQR_regret}, we introduce the notation $\tau_m$ to denote the time step of the start of epoch $m \in \{0,\dots,M-1\}$ with $\tau_0 = 0$, where the word `epoch' refers to the set of time indices for which the controller $K_t$ has been kept fixed because the condition in line~\ref{algstep:info_gain_criterion} of Algorithm~\ref{alg:IR-LQR} has not been satisfied. For simplicity, we assume that $\tau_M = T$, i.e., the experiment is run until the end of epoch $M$, but the arguments presented here can also be generalized to arbitrary $T$. We also define the following quantities:
\begin{align*}
    \bar{w}_t :=\max_{0 \leq k \leq t} \norm{w_k},\, 
    \bar{x}_t :=\max_{0 \leq k \leq t} \norm{x_k},\, 
    \bar{\Sigma} := \max_{K \in \calK} \norm{\Sigma^K_{\Theta_\star}}_2.
\end{align*}

Consider the success event
\begin{align*}
    \calE := \calE^\mathrm{ID} \cap \calE^\mathrm{noise} \cap \calE^\mathrm{conc},
\end{align*}
where
\begin{align*}
    \calE^\mathrm{ID} &:= \curly{\bnorm{\hat{\Theta}_t - \Theta_\star}_{V_t} \leq \beta_t(\delta^\mathrm{ID}), \forall t\geq0}, \\
    \calE^\mathrm{noise}&:= \curly{\bar{w}_{t} \leq 4\sigma_\mathsf{w}\sqrt{\dx\log(\pi^2 (t+1)^3/(6 \delta^{\mathrm{noise}}))}, \, \forall t \geq 0}, \\
    \calE^\mathrm{conc}&:=\bigcap_{m=0}^{M-1}\bigcap_{N=1}^\infty \calE_{m,N}^\mathrm{conc},
\end{align*}
and
\begin{align}
    \calE_{m,N}^\mathrm{conc} := \bigg\{\big\|N \Sigma^{K_{\tau_m}}_{\Theta_\star} &- \sum_{t=\tau_m}^{\tau_m+N-1}x_t^{(m)}x_t^{(m)\top}\big\|_2 \label{eq:conc_mN_def} \\
    &\leq \alpha(\delta_{m,N}, \tau_m,N)\sqrt{N} \bigg\}. \nonumber
\end{align}
Here, 
\begin{align}
    \alpha(\delta, t, N) &:= C^2\left(\frac{\bar{x}_t^2 + \bar{\Sigma}}{1-\rho^2} + \frac{2 \bar{w}_{t+N} \bar{x}_t}{(1-\rho)^2}\right) + (\tilde{\alpha}(\delta)^2 + \tilde{\alpha}(\delta))\bar{\Sigma}, \nonumber \\
    \tilde{\alpha}(\delta) &:= \frac{C}{1-\rho}\left(\sqrt{576\dx\log(18/\delta)} \right), \label{eq: concentration coefficient} 
\end{align}
and $x_t^{(m)}$ denotes the state resulting from applying $K_{\tau_m}$ initialized at $x_{\tau_m}^{(m)}=x_{\tau_m}$.

\begin{lemma} \label{lem:conc_m}
    Let $\delta_{m,N} \in (0,1)$. Then $\mathbb{P}[\calE^\mathrm{conc}_{m,N}] \geq 1-\delta_{m,N}$.
\end{lemma}
\begin{proof}
    For $N=1$, the desired inequality holds deterministically: $\snorm{\Sigma_{\Theta_\star}^{K_{\tau_m}}-x_{\tau_m}x_{\tau_m}^\top}_2 \leq \bar{\Sigma} + \bar{x}_{\tau_m}^2 \leq \alpha(\delta_{m,1},\tau_m,1)$. Assume henceforth $N\geq2$. For notational simplicity, we introduce the shifted time indexing $\tilde{x}_i := x_{\tau_m + i}^{(m)}$ and let
    $\calA := (A_\star  + B_\star K_{\tau_m})$
    . We analyze the concentration inequality by splitting it according to
    \begin{align*}
        \bnorm{N \Sigma^{K_{\tau_m}}_{\Theta_\star} - \sum_{i=0}^{N-1}\tilde{x}_i\tilde{x}_i^\top}_2 \leq \underbrace{\bnorm{N \Sigma^{K_{\tau_m}}_{\Theta_\star} - \sum_{i=0}^{N-1}\mathbb{E}[\tilde{x}_i\tilde{x}_i^\top|\tilde{x}_0]}_2}_{\heartsuit} \\ + \underbrace{\bnorm{\sum_{i=0}^{N-1}\mathbb{E}[\tilde{x}_i\tilde{x}_i^\top|\tilde{x}_0] - \sum_{i=0}^{N-1}\tilde{x}_i\tilde{x}_i^\top}_2}_{\clubsuit}.
    \end{align*}
    \begin{enumerate}
        \item First, we bound $\heartsuit$. From
        \begin{align*}
            \tilde{x}_i =\calA^i\tilde{x}_0 + \sum_{k=0}^{i-1}\calA^i w_i,
        \end{align*}
        we get that
        \begin{align*}
            &\sum_{i=0}^{N-1}\mathbb{E}[\tilde{x}_i\tilde{x}_i^\top|\tilde{x}_0] - N \Sigma^{K_{\tau_m}}_{\Theta_\star} \\
            &= \sum_{i=0}^{N-1}\calA^i\tilde{x}_0\tilde{x}_0^\top (\calA^i)^\top
            + \sum_{i=0}^{N-1} (\sum_{k=0}^{i-1}\sigma_\mathsf{w}^2\calA^k(\calA^k)^\top - \Sigma^{K_{\tau_m}}_{\Theta_\star}) \\
            &=\sum_{i=0}^{N-1}\calA^i\tilde{x}_0\tilde{x}_0^\top (\calA^i)^\top
            - \sum_{i=0}^{N-1} \sum_{k=i}^{\infty}\sigma_\mathsf{w}^2\calA^k(\calA^k)^\top \\
            &=\sum_{i=0}^{N-1}\calA^i\tilde{x}_0\tilde{x}_0^\top (\calA^i)^\top
            - \sum_{i=0}^{N-1} \calA^i\sum_{k=0}^{\infty}\sigma_\mathsf{w}^2\calA^k(\calA^k)^\top (\calA^i)^\top \\
            &=\sum_{i=0}^{N-1}\calA^i\tilde{x}_0\tilde{x}_0^\top (\calA^i)^\top
            - \sum_{i=0}^{N-1} \calA^i \Sigma^{K_{\tau_m}}_{\Theta_\star} (\calA^i)^\top,
        \end{align*}
        resulting in
        \begin{align}
            \heartsuit &\leq (\norm{\tilde{x}_0}^2 + \bnorm{\Sigma^{K_{\tau_m}}_{\Theta_\star}}_2)\sum_{i=0}^\infty \norm{\calA^i}_2^2 \nonumber\\
            &\leq \frac{C^2}{1-\rho^2}(\bar{x}_{\tau_m}^2 + \bar{\Sigma}) \label{eq:proof:heartbound}.
        \end{align}
        \item As for $\clubsuit$, define 
        \begin{align*}
            \calM &:= (\sum_{i=1}^{N-1}\sum_{k=0}^{i-1}\sigma_\mathsf{w}^2\calA^k (\calA^k)^\top)^{-1/2}, \\
            \calW^\top \calW &:= \sum_{i=1}^{N-1}\left(\sum_{k=0}^{i-1}\calA^kw_k\right) \left(\sum_{k=0}^{i-1}\calA^kw_k\right)^\top, \\
            \calL &:= \begin{bmatrix}
                I & 0 & \dots & 0 \\
                \calA & I & \dots & 0 \\
                \vdots & & \ddots & \\
                \calA^{N-1} & \calA^{N-2} & \dots & I
                \end{bmatrix}
                \in \mathbb{R}^{N \dx \times N \dx}.
        \end{align*}
        Then,
        \begin{align*}
            &\clubsuit \leq \norm{\calW^\top\calW - \calM^{-2}}_2 \\
            &+ \underbrace{\norm{\sum_{i=1}^{N-1}\left(\sum_{k=0}^{i-1}\calA^kw_k\right)\tilde{x}_0^\top (\calA^i)^\top + \calA^i\tilde{x}_0\left(\sum_{k=0}^{i-1}\calA^kw_k\right)^\top}_2}_{\blacklozenge},
        \end{align*}
        where
        \begin{align} 
            \blacklozenge &\leq 2 \bar{w}_{\tau_m+N}\norm{\tilde{x}_0}\sum_{i=1}^{N-1} \norm{\calA^i} \sum_{k=0}^{i-1} \norm{\calA^k} \nonumber\\
            &\leq \frac{C^2}{(1-\rho)^2} 2 \bar{w}_{\tau_m+N} \bar{x}_{\tau_m}.\label{eq:proof:conc_bound2}
        \end{align}
        
        In order to bound $\norm{\calW^\top\calW - \calM^{-2}}_2$, we define
        \begin{align*}
            \varepsilon(\delta) := \frac{1}{\sqrt{N-1}}\tilde{\alpha}(\delta),
        \end{align*}
        in which case
        \begin{align*}
            \exp\left(\dx\log(18)-\frac{\varepsilon^2(\delta)}{576\sigma_\mathsf{w}^2\norm{\calM}_2^2\norm{\calL}_2^2}\right) \leq \delta,
        \end{align*}
        due to $\norm{\calL}_2^2 \leq C^2/(1-\rho)^2$ and 
         \begin{align*}
            \norm{\calM}^2_2 &= 1/\lambda_{\min}(\sum_{i=1}^{N-1}\sum_{k=0}^{i-1}\sigma_\mathsf{w}^2\calA^k (\calA^k)^\top) \\
            &\leq 1/((N-1)\sigma_\mathsf{w}^2).
        \end{align*}
        Applying~\cite[Lemma~B.2]{ziemann2023tutorial} results in
        \begin{align} \label{eq:LemmaB1}
            \mathbb{P}[\norm{\calM\calW^\top\calW\calM - I}_2 > \max(\varepsilon(\delta),\varepsilon^2(\delta))] \leq \delta. 
        \end{align}
        Since
        \begin{align*}
            \norm{\calM^{-2}}_2 &= \norm{\sum_{i=1}^{N-1}\sum_{k=0}^{i-1}\sigma_\mathsf{w}^2\calA^k (\calA^k)^\top}_2 \\
            &\leq (N-1)\bnorm{\Sigma^{K_{\tau_m}}_{\Theta_\star}}_2,
        \end{align*}
        we have that
        \begin{align*}
            \big\|\calW^\top\calW &- \calM^{-2}\big\|_2 = \norm{\calM^{-1}(\calM\calW^\top\calW\calM - I)\calM^{-1}}_2 \\
            &\leq \norm{\calM^{-2}}_2 \norm{\calM\calW^\top\calW\calM - I}_2 \\
            &\leq (N-1)\bnorm{\Sigma^{K_{\tau_m}}_{\Theta_\star}}_2 \norm{\calM\calW^\top\calW\calM - I}_2.
        \end{align*}
        This, combined with~\eqref{eq:LemmaB1}, further implies that
        \begin{align}
            &\norm{\calW^\top\calW - \calM^{-2}}_2 \nonumber\\
            &\quad \leq \max(\varepsilon(\delta_{m,N}),\varepsilon^2(\delta_{m,N}))(N-1)\bnorm{\Sigma^{K_{\tau_m}}_{\Theta_\star}}_2 \nonumber\\
            &\quad \leq (\tilde{\alpha}(\delta_{m,N})^2 + \tilde{\alpha}(\delta_{m,N}))\bar{\Sigma}\sqrt{N}, \label{eq:proof:conc_bound1}
        \end{align}
        with probability at least $1-\delta_{m,N}$.
    \end{enumerate}
    Finally, combining~\eqref{eq:proof:heartbound}, \eqref{eq:proof:conc_bound2}, and~\eqref{eq:proof:conc_bound1} concludes the proof.
\end{proof}

Armed with this result, we can state that the success event $\calE$ occurs with high probability.
\begin{lemma} \label{lem:high_prob}
    Let $\delta \in (0,1)$ and select $\delta_\mathrm{ID} = \delta_\mathrm{noise} = \delta/3$ and $\delta_{m,N}= 12\delta/(\pi^4 (m+1)^2 N^2)$. Then
    \begin{align*}
        \mathbb{P}[\calE] \geq 1 - \delta.
    \end{align*}
\end{lemma}
\begin{proof}
    We can prove this result by considering the failure probabilities of the individual events and using the union bound to obtain the overall success probability. Namely:
    \begin{enumerate}
        \item $\mathbb{P}[\overline{\calE^{\mathrm{noise}}}] \leq \delta^{\mathrm{noise}} = \delta/3$ holds by union bounding over the event for $t \geq 0$ from~\cite[Lemma 13]{lee2024nonasymptotic} with $\delta \gets 6 \delta^{\mathsf{noise}}/(\pi^2 (t+1)^2)$. 
        \item Due to Lemma~\ref{lem:id_success_event}, it holds that $\mathbb{P}[\overline{\calE^\mathrm{ID}}] \leq \delta^\mathrm{ID} = \delta/3$.
        \item Due to Lemma~\ref{lem:conc_m}, it holds that $\mathbb{P}[\overline{\calE^\mathrm{conc}_{m,N}}] \leq \delta_{m,N}, \forall m \in [M-1], \forall N \geq 1$. Union bounding this yields
        \begin{align*}
            \mathbb{P}&\left[\, \overline{\bigcap_{m=0}^{M-1}\bigcap_{N=1}^\infty \calE_{m,N}^\mathrm{conc}}\, \right] = \mathbb{P}\left[ \bigcup_{m=0}^{M-1} \bigcup_{N=1}^\infty \overline{\calE_{m,N}^\mathrm{conc}}\right] \\
            &\leq \sum_{m=0}^{M-1} \sum_{N=1}^\infty\mathbb{P}\left[ \overline{\calE_{m,N}^\mathrm{conc}}\right] \leq \sum_{m=0}^{M-1} \sum_{N=1}^\infty \delta_{m,N} \\
            &\leq \frac{12\delta}{\pi^4}\sum_{m=0}^{\infty}\sum_{N=1}^\infty \frac{1}{(m+1)^2N^2} = \frac{\delta}{3},
        \end{align*}
        i.e., $\mathbb{P}[\overline{\calE^\mathrm{conc}}] \leq \delta/3$. Consequently,
        \begin{align*}
            \mathbb{P}[\overline{\calE}] &= \mathbb{P} \left[ \overline{\calE^\mathrm{ID}} \cup \overline{\calE^\mathrm{noise}} \cup \overline{\calE^\mathrm{conc}} \,\right] \\
            &\leq \mathbb{P}[\overline{\calE^\mathrm{ID}}] +\mathbb{P}[\overline{\calE^\mathrm{noise}}] + \mathbb{P}[\overline{\calE^\mathrm{conc}}] \leq \delta. \qedhere
        \end{align*}
    \end{enumerate}
\end{proof}

\subsection{Proof of Theorem~\ref{thm:IR-LQR_regret}} \label{sec:app:main}
Having defined the success events in Appendix~\ref{sec:app:events}, we can proceed with the proof of the sublinear regret result of Theorem~\ref{thm:IR-LQR_regret}. Lemmas~\ref{lem:bounded_state} and~\ref{lem:sum_bound} will be useful to complete the proof:

\begin{lemma} \label{lem:bounded_state}
Consider the sequence of controllers $\{K_{\tau_m}\}_{m=0}^{M-1} \subseteq \mathcal{K}$,
where $K_{\tau_m}$ is applied in epoch $m$ for a total of
$T = \sum_{m=0}^{M-1}(\tau_{m+1}-\tau_m)=\tau_M$ time steps,
starting from $x_0=0$. Define $\tau_{\mathrm{stab}}:=\left\lceil\log(2C)/\log(1/\rho)\right\rceil.$
    
If $\tau_{m+1}-\tau_m \geq \tau_{\mathrm{stab}}$ for all
$m=0,\ldots,M-1$, then
\begin{align*}
    \bar{x}_{\tau_{m+1}}
    \leq
    c_x \bar{w}_{\tau_{m+1}},
    \qquad
    m=0,\ldots,M-1,
\end{align*}
where $c_x:=C(2C+1)/(1-\rho)$.
\end{lemma}

\begin{proof}
Let $A_m := A_\star + B_\star K_{\tau_m}.$ By Lemma~\ref{lem:every_theta_stable}, $\|A_m^k\|_2 \leq C\rho^k$ for all $k\geq 0$. Hence, for
$s\in\{0,\ldots,\tau_{m+1}-\tau_m\}$,
\begin{align*}
    \|x_{\tau_m+s}\|
    &\leq
    \|A_m^s\|_2\|x_{\tau_m}\|
    +
    \sum_{j=0}^{s-1}
    \|A_m^{s-1-j}\|_2\|w_{\tau_m+j}\| \\
    &\leq
    C\rho^s\|x_{\tau_m}\|
    +
    \frac{C}{1-\rho}\bar{w}_{\tau_{m+1}}.
\end{align*}
In particular, since $C\rho^{\tau_{\mathrm{stab}}}\leq \frac{1}{2},$
it follows that
\begin{align*}
    \|x_{\tau_{m+1}}\|
    \leq
    \frac{1}{2}\|x_{\tau_m}\|
    +
    \frac{C}{1-\rho}\bar{w}_{\tau_{m+1}}.
\end{align*}
Since $x_0=0$ and $\bar{w}_t$ is nondecreasing, induction over the
epochs yields
\begin{align*}
    \|x_{\tau_m}\|
    \leq
    \frac{2C}{1-\rho}\bar{w}_{\tau_m},
    \qquad
    m=0,\ldots,M.
\end{align*}
Substituting this bound above and using
$\bar{w}_{\tau_m}\leq\bar{w}_{\tau_{m+1}}$ gives
\begin{align*}
    \|x_{\tau_m+s}\|
    \leq
    \left(
        \frac{2C^2}{1-\rho}
        +
        \frac{C}{1-\rho}
    \right)
    \bar{w}_{\tau_{m+1}} =
    \frac{C(2C+1)}{1-\rho}
    \bar{w}_{\tau_{m+1}}
\end{align*}
for all $s\in\{0,\ldots,\tau_{m+1}-\tau_m\}$. Taking the maximum
over all time steps up to $\tau_{m+1}$ concludes the proof.
\end{proof}

Under event $\mathcal{E}^{\mathrm{noise}}$, define
\begin{align*}
    \hat{w}_t
    &:=
    4\sigma_\mathsf{w}
    \sqrt{
        \dx\log\left(
            \frac{\pi^2(t+1)^3}{6\delta^{\mathrm{noise}}}
        \right)
    },
\end{align*}
such that $\bar{w}_t \leq \hat{w}_t$ for all $t\geq 0$.
Furthermore, for $m\geq 0$ and $t\geq \tau_m$, define
\begin{align*}
    \underline{\delta}_{m,t}
    &:=
    \frac{12\delta}
    {\pi^4(m+1)^2(t+1)^2}.
\end{align*}
For $N=t-\tau_m\geq 1$, it holds that
$\underline{\delta}_{m,t}\leq\delta_{m,N}$. Using
Lemma~\ref{lem:bounded_state} and the choice of $\delta_{m,N}$ in Lemma~\ref{lem:high_prob},
we can therefore define the following bound on $\alpha$ which does not depend on the realization of the noise:
\begin{align}
    \bar{\alpha}(m,t)
    &:=
    C^2\left(
        \frac{
            c_x^2 \hat{w}_t^2+\bar{\Sigma}
        }{1-\rho^2}
        +
        \frac{
            2c_x \hat{w}_t^2
        }{(1-\rho)^2}
    \right)
    \nonumber\\
    &\quad+
    \left(
        \tilde{\alpha}(\underline{\delta}_{m,t})^2
        +
        \tilde{\alpha}(\underline{\delta}_{m,t})
    \right)\bar{\Sigma},
\end{align}
which can be used to compute the epoch lower bound
\begin{align}\label{eq:tau_underline_def}
    \underline \tau(m,t)
    &:=
    \max\left\{
        \tau_{\mathrm{stab}},
        \frac{4\bar{\alpha}(m,t)^2}{\sigma_\mathsf{w}^4}
    \right\}.
\end{align}

An immediate consequence of Lemma~\ref{lem:bounded_state} is the following result.
\begin{corollary} \label{cor:bounded_state}
    Under event $\calE^\mathrm{noise}$ and the same conditions as in Lemma~\ref{lem:bounded_state}, it holds that 
    \begin{enumerate}
        \item $\bar \alpha(m, \tau_{m+1}) \geq \alpha(\delta_{m,\tau_{m+1}-\tau_m}, \tau_m,\tau_{m+1}-\tau_m)$
        \item $\bar{x}_T = \mathcal{O}(\sqrt{\log T})$,
        \item $\bar{\alpha}(M,T) = \calO(\mathsf{poly}\log T)$,
        \item $\norm{V_T}_2 = \tilde{\calO}(T)$,
        \item $\beta_T = \calO(\sqrt{\log T})$,
        \item $g_T = \tilde{\calO}(\sqrt{T})$.
    \end{enumerate}
\end{corollary}

\begin{lemma} \label{lem:sum_bound}
    Under event $\calE$, it holds that
    \begin{align*}
        \sum_{m=0}^{M-1} \sum_{t=\tau_m}^{\tau_{m+1}-1}  a_t\trace\left( V_{\tau_m}^{-1} \tilde{\Sigma}^{K_{\tau_m}}_{\Theta_\star}\right) = \tilde{\calO}(\bar{a}_T),
    \end{align*}
    for a nonnegative scalar sequence $\{a_t\}_{t=0}^{T}$ and $\bar{a}_T:=\max_{0\leq t\leq T}a_t$.
\end{lemma}
\begin{proof}
    Let $E_m = \sum_{t=\tau_m}^{\tau_{m+1}-1}(\tilde{\Sigma}^{K_{\tau_m}}_{\Theta_\star} - x_tx_t^\top)$. Under event $\calE^\mathrm{conc}_{m,\tau_{m+1}-\tau_m} \cap \calE^\mathrm{noise}$,
    Corollary~\ref{cor:bounded_state} implies that $\norm{E_m}_2 \leq \bar \alpha(m, \tau_{m+1}) \sqrt{\tau_{m+1}-\tau_m}$. Then by the fact that $\tau_{m+1} - \tau_m \geq \underline \tau(m, \tau_{m+1}) \geq 4 \frac{\bar \alpha(m, \tau_{m+1})^2}{\sigma_\mathsf{w}^4}$, it holds that
    \begin{align*}
        \sum_{t=\tau_m}^{\tau_{m+1}-1} x_t x_t^\top \succeq \frac{\sigma_\mathsf{w}^2(\tau_{m+1} - \tau_m)}{2}.
    \end{align*}
    Therefore
    \begin{align*}
        &\trace\left(V_{\tau_m}^{-1}\bmat{I \\ K_{\tau_m}} E_m \bmat{I \\ K_{\tau_m}}^\top\right) \\
        &\leq \bar \alpha(m, \tau_{m+1})\sqrt{\tau_{m+1}-\tau_m} \trace\left(V_{\tau_m}^{-1}\bmat{I \\ K_{\tau_m}} \bmat{I \\ K_{\tau_m}}^\top\right) \\
        &\leq\frac{2 \bar \alpha(m, \tau_{m+1})}{\sigma_\mathsf{w}^2\sqrt{\tau_{m+1}-\tau_m}}\trace\left(V_{\tau_m}^{-1}\bmat{I \\ K_{\tau_m}} \sum_{t=\tau_m}^{\tau_{m+1}-1}x_tx_t^\top\bmat{I \\ K_{\tau_m}}^\top\right).
    \end{align*}
  Then it follows that
    \begin{align*}
        &(\tau_{m+1}-\tau_m)\trace(V_{\tau_m}^{-1} \tilde{\Sigma}^{K_{\tau_m}}_{\Theta_\star}) \\
        &\leq \trace\left(V_{\tau_m}^{-1}\bmat{I \\ K_{\tau_m}} E_m \bmat{I \\ K_{\tau_m}}^\top\right) \\
        &+ \trace\left(V_{\tau_m}^{-1}\bmat{I \\ K_{\tau_m}} \sum_{t=\tau_m}^{\tau_{m+1}-1}x_tx_t^\top\bmat{I \\ K_{\tau_m}}^\top\right) \\
        &\leq \left(1+\frac{2\bar \alpha(m, \tau_{m+1})}{\sigma_\mathsf{w}^2\sqrt{\underline{\tau}(m, \tau_{m+1})}}\right) \\
        &\times \trace\left(V_{\tau_m}^{-1}\bmat{I \\ K_{\tau_m}} \sum_{t=\tau_m}^{\tau_{m+1}-1}x_tx_t^\top\bmat{I \\ K_{\tau_m}}^\top\right).
    \end{align*}
    Consequently,
    \begin{align*}
        &\sum_{m=0}^{M-1} \sum_{t=\tau_m}^{\tau_{m+1}-1}  a_t\trace\left( V_{\tau_m}^{-1} \tilde{\Sigma}^{K_{\tau_m}}_{\Theta_\star}\right) \\
        &\leq \bar{a}_T \sum_{m=0}^{M-1} (\tau_{m+1}-\tau_m)\trace(V_{\tau_m}^{-1} \tilde{\Sigma}^{K_{\tau_m}}_{\Theta_\star})\\
        &\leq\left(1+\frac{2\bar \alpha(M, T)}{\sigma_\mathsf{w}^2\sqrt{\underline{\tau}(M, T)}}\right)\bar{a}_T \times \\
        &\quad\sum_{m=0}^{M-1}\sum_{t=\tau_m}^{\tau_{m+1}-1}\trace\left(V_{\tau_m}^{-1}\bmat{I \\ K_{\tau_m}}x_tx_t^\top\bmat{I \\ K_{\tau_m}}^\top\right),
    \end{align*}
    where the last inequality holds by the monotonicity of the sequence $\underline{\tau}(m, t)$ and $\bar \alpha(m, t)$ as $m$ and $t$ increase. 
    Furthermore,
    \begin{align*}
        &\sum_{m=0}^{M-1}\sum_{t=\tau_m}^{\tau_{m+1}-1}\trace\left(V_{\tau_m}^{-1}\bmat{I \\ K_{\tau_m}}x_tx_t^\top\bmat{I \\ K_{\tau_m}}^\top\right) \\
        &= \sum_{m=0}^{M-1}\sum_{t=\tau_m}^{\tau_{m+1}-1}\trace\left(V_{\tau_m}^{-1}\bmat{x_t \\ u_t} \bmat{x_t \\ u_t}^\top\right) \\
        &= \sum_{m=0}^{M-1}\sum_{t=\tau_m}^{\tau_{m+1}-1} \norm{V_{\tau_m}^{-1/2} \bmat{x_t \\ u_t}}_2^2 \\
        &= \sum_{m=0}^{M-1}\sum_{t=\tau_m}^{\tau_{m+1}-1} \norm{V_{\tau_m}^{-1/2}V_t^{1/2}V_t^{-1/2} \bmat{x_t \\ u_t}}_2^2 \\
        &\leq \sum_{m=0}^{M-1}\sum_{t=\tau_m}^{\tau_{m+1}-1} \norm{V_{\tau_m}^{-1/2}V_t^{1/2}}_2^2 \norm{V_t^{-1/2} \bmat{x_t \\ u_t}}_2^2 \\
        &\leq \sum_{m=0}^{M-1}\sum_{t=\tau_m}^{\tau_{m+1}-1} \max\curly{2, \underline{\tau}(m, \tau_{m+1}))
        \bar x_T^2/\lambda}\norm{V_t^{-1/2} \bmat{x_t \\ u_t}}_2^2 \\
        &= \calO(\mathsf{poly}\log T).
    \end{align*}
    The last inequality follows from $\det(V_{t-1}) \leq 2 \det(V_{\tau_m})$ and $V_{t-1} \succeq V_{\tau_m}$ along with a crude bound on the change in the covariance over an epoch in the event that the lower bound $\tau_{m+1} - \tau_m \geq \underline{\tau}(m, \tau_{m+1})$ is active. 
    The last equality is a standard result based on the maximum information gain associated with linear kernels~\cite{srinivas2009gaussian} along with bounds on $\bar x_T$ and $\underline \tau(m, \tau_{m+1})$. Combining this result with Corollary~\ref{cor:bounded_state} concludes the proof.
\end{proof}

The last part of the proof of Theorem~\ref{thm:IR-LQR_regret} is to show that as long as the success event $\calE$ holds, the regret is of order $\tilde{\calO}(\sqrt{T})$. Defining
\begin{align*}
    \Delta_m = \sum_{t = \tau_m}^{\tau_{m+1}-1}\ell(x_t,u_t) - (\tau_{m+1} - \tau_m) J(K_{\tau_m}, \Theta_\star),
\end{align*}
allows us to rewrite the cumulative regret as
\begin{align*}
    \calR_T &= \sum_{m=0}^{M-1} \left( \sum_{t = \tau_m}^{\tau_{m+1}-1} \ell(x_t,u_t) - (\tau_{m+1} - \tau_m) J(K_\star, \Theta_\star) \right) \\
    &= \sum_{m=0}^{M-1} \left( \Delta_m + (\tau_{m+1} - \tau_m)(J(K_{\tau_m}, \Theta_\star) - J(K_\star, \Theta_\star) )\right),
\end{align*}
where
\begin{align*}
    J&(K_{\tau_m}, \Theta_\star) - J(K_\star, \Theta_\star) \\
    &\stackrel{\eqref{eq:Jbar_optimistic}}{\leq} J(K_{\tau_m}, \Theta_{\star}) - \bar{J}(K_\star, \hat{\Theta}_{\tau_m},\calV_{\tau_m}) \\
    &\leq J(K_{\tau_m}, \Theta_\star) - \bar{J}(K_{\tau_m}, \hat{\Theta}_{\tau_m}\calV_{\tau_m}) \\
    &\leq J(K_{\tau_m}, \Theta_\star) - J(K_{\tau_m}, \hat{\Theta}_{\tau_m}) + \calB(K_{\tau_m},\hat{\Theta}_{\tau_m},\calV_{\tau_m}) \\
    &\stackrel{\eqref{eq:proof:optimism_theta_hat}}{\leq} 2 \calB(K_{\tau_m},\hat{\Theta}_{\tau_m},\calV_{\tau_m}) \\
    &\stackrel{\eqref{eq:calB}}{\leq} 2\trace(\calV_{\tau_m}\tilde{\Sigma}_{\hat{\Theta}_{\tau_m}}^{K_{\tau_m}}).
\end{align*}
Therefore, the regret can be decomposed as
\begin{align*}
    \calR_T \leq \calR_T^{(1)} + 2\calR_T^{(2)},
\end{align*}
with
\begin{align*}
    \calR^{(1)}_T &= \sum_{m=0}^{M-1} \Delta_m, \\
    \calR^{(2)}_T &= \sum_{m=0}^{M-1}  (\tau_{m+1} - \tau_m)\trace(\calV_{\tau_m}\tilde{\Sigma}^{K_{\tau_m}}_{\hat{\Theta}_{\tau_m}}).
\end{align*}

In the following, we bound the two regret components.
\begin{enumerate}
    \item Bounding $\calR^{(1)}_T$.
    \begin{align*}
        &\Delta_m = - (\tau_{m+1}-\tau_m)J(K_{\tau_m},\Theta_\star) + \sum_{t=\tau_m}^{\tau_{m+1}-1} \ell(x_t,u_t) \\
        &= - (\tau_{m+1}-\tau_m)J(K_{\tau_m},\Theta_\star)\\
        &+ \sum_{t=\tau_m}^{\tau_{m+1}-1} \trace\left(\bmat{ Q & \\ & R } \bmat{I \\ K_{\tau_m}} x_t x_t^\top \bmat{ I \\ K_{\tau_m}}^\top\right) \\
        &= - (\tau_{m+1}-\tau_m)J(K_{\tau_m},\Theta_\star) \\
        &\hspace{-1em}+ \trace\left((Q+K_{\tau_m}^\top R K_{\tau_m}) \sum_{t=\tau_m}^{\tau_{m+1}-1}(x_t x_t^\top-\Sigma^{K_{\tau_m}}_{\Theta_\star} + \Sigma^{K_{\tau_m}}_{\Theta_\star})\right) \\
        &\leq - (\tau_{m+1}-\tau_m)J(K_{\tau_m},\Theta_\star) \\
        &+ \sum_{t=\tau_m}^{\tau_{m+1}-1}\trace(\diag(Q,R)\tilde{\Sigma}^{K_{\tau_m}}_{\Theta_\star}) \\
        &+ \fnorm{(Q+K_{\tau_m}^\top R K_{\tau_m})} \fnorm{\sum_{t=\tau_m}^{\tau_{m+1}-1}(x_t x_t^\top - \Sigma_{\Theta_\star}^{K_{\tau_m}})} \\
        &\leq (\fnorm{Q}+\fnorm{R}\bar{K}^2) \alpha(\delta_m,\tau_{m+1})\sqrt{\dx}\sqrt{\tau_{m+1}-\tau_m},
    \end{align*}
    where the last inequality is due to the covariance concentration event $\calE^\mathrm{conc}_{m,N}$. Combined with Corollary~\ref{cor:bounded_state}, this implies that $\Delta_m =  \tilde{\calO}(\sqrt{\tau_{m+1}-\tau_m})$, furthermore, $\sum_{m=0}^{M-1}\sqrt{\tau_{m+1}-\tau_m} \leq \sqrt{M} \sqrt{T}$. Due to the information gain criterion in line~\ref{algstep:info_gain_criterion} of Algorithm~\ref{alg:IR-LQR}, $\det(V_T) = \det(V_{\tau_M}) \geq 2^M \det(V_0) = 2^M \det(\lambda I)$, therefore
    \begin{align*}
        (2^{1/(\dx+\du)})^M \lambda &\leq \det(V_T)^{1/(\dx+\du)} \\ &\leq \norm{V_T}_2 \leq (1+\bar{K}^2)\bar{x}_T^2 T + \lambda,
    \end{align*}
    resulting in $M=\calO(\log T)$. Consequently, $\calR^{(1)}_T = \tilde{\calO}(\sqrt{T})$.
    \item Bounding $\calR^{(2)}_T$. Note that
    \begin{align*}
        \trace(V_{\tau_m}^{-1}\tilde{\Sigma}^{K_{\tau_m}}_{\hat{\Theta}_{\tau_m}}) = \trace(V_{\tau_m}^{-1}(\tilde{\Sigma}^{K_{\tau_m}}_{\hat{\Theta}_{\tau_m}}-\tilde{\Sigma}^{K_{\tau_m}}_{\Theta_\star})) + \trace(V_{\tau_m}^{-1}\tilde{\Sigma}^{K_{\tau_m}}_{\Theta_\star}),
    \end{align*}
    where
    \begin{align*}
        &\trace(V_{\tau_m}^{-1}(\tilde{\Sigma}^{K_{\tau_m}}_{\hat{\Theta}_{\tau_m}}-\tilde{\Sigma}^{K_{\tau_m}}_{\Theta_\star})) \\
        &\leq \bfnorm{\Sigma^{K_{\tau_m}}_{\hat{\Theta}_{\tau_m}}-\Sigma^{K_{\tau_m}}_{\Theta_\star}}\fnorm{\bmat{I \\ K_{\tau_m}}^\top V_{\tau_m}^{-1} \bmat{I \\ K_{\tau_m}}}.
    \end{align*}
    Using similar arguments as in the proof of Lemma~\ref{lem:exploration_bounds_cost_diff}, we get that
    \begin{align*}
        &\bfnorm{\Sigma^K_{\Theta_\star} -  \Sigma^K_{\hat{\Theta}_t}} \\ &\leq \frac{C^2}{1-\rho^2}\Big( 2 \bfnorm{\hat{\Theta}_t}  \bfnorm{\tilde{\Sigma}^K_{\Theta_\star} (\Theta^\star-\hat{\Theta}_t)^\top} \\
        &+ \bfnorm{ (\Theta^\star-\hat{\Theta}_t) \tilde{\Sigma}^K_{\Theta_\star} (\Theta^\star-\hat{\Theta}_t)^\top}\Big) \\
        &\leq \frac{1}{\lambda}\left(\bar{\Theta} + \frac{1}{\lambda}\right)\frac{4C^2}{1-\rho^2} \trace(\Sigma^{K_{\tau_m}}_{\Theta_\star}) \\
        &= c(\lambda)\trace(\tilde{\Sigma}^{K_{\tau_m}}_{\Theta_\star})/(\fnorm{Q}+\bar{K}^2\fnorm{R}).
    \end{align*}
    Furthermore, 
    \begin{align*}
        \fnorm{\bmat{I \\ K_{\tau_m}}^\top V_{\tau_m}^{-1} \bmat{I \\ K_{\tau_m}}} &\leq \trace\left(\bmat{I \\ K_{\tau_m}}^\top V_{\tau_m}^{-1} \bmat{I \\ K_{\tau_m}} \right) \\
        &\leq \frac{1}{\sigma_\mathsf{w}^2} \trace(V_{\tau_m}^{-1}\tilde{\Sigma}^{K_{\tau_m}}_{\Theta_\star}).
    \end{align*}
    Thus,
    \begin{align*}
        &\trace(\calV_{\tau_m}\tilde{\Sigma}^{K_{\tau_m}}_{\hat{\Theta}_{\tau_m}}) \leq g_{\tau_m}\trace(V_{\tau_m}^{-1}\tilde{\Sigma}^{K_{\tau_m}}_{\hat{\Theta}_{\tau_m}}) \\
        &\leq
        g_{\tau_m}\left(1 + \frac{c(\lambda)\trace(\tilde{\Sigma}^{K_{\tau_m}}_{\Theta_\star})}{\sigma_\mathsf{w}^2(\fnorm{Q}+\bar{K}^2\fnorm{R})} \right)\trace(V_{\tau_m}^{-1}\tilde{\Sigma}^{K_{\tau_m}}_{\Theta_\star}).
    \end{align*}
    Lemma~\ref{lem:sum_bound} together with Corollary~\ref{cor:bounded_state} results in $\calR^{(2)}_T = \tilde{\calO}(g_T) = \tilde{\calO}(\sqrt{T})$.
\end{enumerate}

In summary, we have shown that event $\calE$ occurs with probability at least $1-\delta$, and the cumulative regret is $\calR_T \leq \calR^{(1)}_T + 2\calR^{(2)}_T = \tilde{\calO}(\sqrt{T})$ under event $\calE$. \hfill $\blacksquare$

\end{document}